\documentclass[a4paper,11pt]{article}
\usepackage{amsmath, amsthm, amssymb}
\usepackage{tabularx}
\usepackage{algorithm}
\usepackage{algpseudocode}
\usepackage{tikz}
\usetikzlibrary {arrows.meta,automata,positioning}
\usepackage{hyperref}
\hypersetup{breaklinks=True}
\usepackage{xurl}
\usepackage{geometry}

\makeatletter
\newcommand{\multiline}[1]{%
  \begin{tabularx}{\dimexpr\linewidth-\ALG@thistlm}[t]{@{}X@{}}
    #1
  \end{tabularx}
}
\makeatother

\newtheorem{theorem}{Theorem}[section]

\newtheorem{lemma}[theorem]{Lemma}

\theoremstyle{definition}
\newtheorem{definition}[theorem]{Definition}
\newtheorem{example}[theorem]{Example}
\newtheorem{remark}[theorem]{Remark}

\title{Pseudo-MDP Convolutional Codes for Burst Erasure Correction}

\author{Zita Abreu, Julia Lieb, Raquel Pinto}

\usepackage{xcolor}

\usepackage{comment}

\begin{document}

\maketitle

\begin{abstract}
Convolutional codes are a class of error-correcting codes that performs very well over erasure channels with low delay requirements. In particular,  Maximum Distance Profile (MDP) convolutional codes, which are defined to have optimal column distances, are able to correct a maximal number of erasures in decoding windows of fixed sizes. However, the required field size in the known constructions for MDP convolutional codes increases rapidly with the code parameters. On the other hand, if the code parameters are small, larger bursts of erasures cannot be corrected.
In this paper, we present a new class of convolutional codes, which we call Pseudo-MDP convolutional codes. By definition these codes can correct large bursts of erasures within a prescribed time-delay and still keep part of the advantageous properties of MDP convolutional codes, in the sense that we require some but not all column distances to be optimal. This release in the condition on the column distances allows us to construct Pseudo-MDP convolutional codes over fields of smaller size than those required for MDP convolutional codes with the same code parameters.

\end{abstract}

\section{Introduction}
Maximum Distance Profile (MDP) Convolutional Codes are a class of codes that present very good performance in error correction when data is transmitted over an erasure channel. Unlike traditional communication channels where data may be corrupted or altered, an erasure channel presents a unique scenario where data is either successfully received intact or completely lost. Erasure channels are directly applicable in real-world systems. They are used to model packet loss in data networks, such as the Internet, where data packets may be dropped and marked as missing. They are also relevant in wireless communication, satellite links, and distributed storage systems where data may be lost due to interference, signal fading, or hardware failures. In many practical communication systems, erasures do not occur randomly or independently; instead, they often appear in clusters, known as bursts of erasures. This phenomenon is particularly common in real-world erasure channels, and it is driven by several underlying physical and network-related factors. Convolutional codes are very well suited for correction over these channels, particularly when the erasures occur in bursts.

MDP convolutional codes are defined by their optimal column distances, which means that  within any window of a certain length, the code maximizes the number of erasures that can be corrected.
Despite their attractive error-correction capabilities, there are relatively few constructions of MDP convolutional codes, and the general constructions that exist tend to require large finite fields (see \cite{GRS:sMDS}, \cite{ANP:MDP}, \cite{NS:MDP}, \cite{L:cMDP}, \cite{rscc}). Recently, there has been growing interest in constructing MDP convolutional codes over smaller fields to make them more practical for implementation in real systems (see  \cite{AL:jMDP}, \cite{MUNOZCASTANEDA}, \cite{unitmemorymdp}, \cite{chen}). In \cite{cheng2026new,chen,itw,unitmemorymdp} code parameters are considered that imply that the column distances can only be optimal until the first column distance, and in \cite{AL:jMDP}, so-called $j$-MDP codes are considered, which are by definition only required to be optimal up to the $j$-th column distance. This is due to the fact that for low-delay decoding especially the small column distances are relevant and if one requires only the first $j+1$ column distances to be optimal, constructions over smaller finite fields are possible. Another drawback of MDP convolutional codes, besides that large finite fields are needed for their construction, is that they do not perform ideally if large bursts of erasures occur. 

In this paper, we propose a new class of $(n,k, \delta)$ convolutional codes, with $k \mid \delta$, which we call Pseudo-MDP convolutional codes and which are defined as codes with optimal first $\frac{\delta}{k}$ column distances that additionally are able to correct bursts of erasures of a prescribed size within a given time delay.
We present constructions of codes with this property
by extending the encoder of an MDP convolutional code $\cal C$ of smaller degree. While the resulting Pseudo-MDP convolutional code may not maintain the MDP property, we keep the field size of the original MDP code even if the newly constructed Pseudo-MDP convolutional code is of larger degree. Since Pseudo-MDP convolutional codes additionally exhibit improved performance in terms of burst erasure correction, these codes offer an alternative to traditional MDP convolutional codes, also because they allow the use of smaller finite fields (when considering the same code parameters).


The paper is organized as follows. In Section 2, we provide preliminary results 
on convolutional codes, with a focus on Maximum Distance Profile (MDP) codes. 
In Section 3, we define and construct $(n,k,\delta)$ Pseudo-MDP convolutional codes where $k | \delta$. 
In Section 4, we generalize the definition of Pseudo-MDP convolutional codes 
by relaxing column distance conditions to support smaller fields and larger 
code rates. Finally, Section 5 provides a brief summary of the results of the paper.

\section{MDP Convolutional Codes}

Throughout this paper, we consider convolutional codes over a finite field $\mathbb{F}_{q}$. Let $\mathbb{F}_{q}[z]$ denote the ring of polynomials in the indeterminate $z$ with coefficients in $\mathbb{F}_{q}$. In this section, we present preliminary results on convolutional codes, in particular, on MDP convolutional codes.

\vspace{0.2cm}

\begin{definition}
 An $(n,k)$ \textbf{convolutional code} is an $\mathbb{F}_{q}[z]$-submodule of $\mathbb{F}_{q}[z]^n$ of rank $k$. A matrix $G(z) \in \mathbb{F}_{q}[z]^{k \times n}$ whose rows form a basis of $\cal C$ is called an \textbf{encoder} of $\cal C$ and we have
\[
\begin{aligned}
\mathcal{C} &= \operatorname{Im}_{\mathbb{F}_q[z]} G(z) \\
           &= \left\{ v(z) \in \mathbb{F}_q[z]^n \, : \, v(z) = u(z) G(z), \, \text{with} \, u(z) \in \mathbb{F}_q[z]^k \right\}.
\end{aligned}
\]
\end{definition}
Moreover $v(z) = u(z)G(z)$ is called the codeword corresponding to the information sequence $u(z)$.

Clearly an $(n,k)$ convolutional code admits multiple encoders and if $G(z)$ and $G'(z)$ are two encoders of $\cal C$ then $$G'(z)=U(z)G(z)$$ for some unimodular matrix $U(z) \in \mathbb{F}_{q}[z]^{k \times k}$ (i.e., a polynomial matrix with inverse in $\mathbb{F}_{q}[z]^{k \times k}$, or equivalently, with determinant in $\mathbb{F}_{q} \backslash \{0\}$). This implies that the full size minors of two encoders of $\cal C$ differ only by a constant factor.

The maximum degree of the full size minors of an encoder of $\cal C$ is called the \textbf{degree} of the code and we denote it by $\delta$. An $(n,k)$ convolutional code with degree $\delta$ is said to be an $(n,k,\delta)$ convolutional code. For $i\in\{1,\hdots,k\}$, the $i$-th \textbf{row degree} of an encoder $G(z)$ is the maximal degree of any entry of row $i$ of $G(z)$. The sum of the row degrees of an encoder is always lower bounded by the degree of the code and if it is equal, $G(z)$ is called \textbf{minimal}. The maximum of the row degrees of $G(z)$ is called the \textbf{memory} $\nu$ of $G(z)$.

\begin{definition}
Let $\mathcal{C}$ be an $(n,k,\delta)$ convolutional code.
If there exists a matrix $H(z)\in\mathbb F_q[z]^{(n-k)\times n}$ of full (row) rank such that $$\mathcal{C}=\{v(z)\in\mathbb F_q[z]^n\ |\ H(z)v(z)^\top=0\},$$ this matrix $H(z)$ is called \textbf{parity-check} matrix for $\mathcal{C}$. Convolutional codes for which there exists a parity-check matrix are called \textbf{non-catastrophic}. Convolutional codes without parity-check matrix are called \textbf{catastrophic}.
\end{definition}

The following theorem gives criteria how to see from an encoder of a convolutional code whether the code is non-catastrophic.


\begin{theorem}[\cite{York97},\cite{kailath1980}] \label{Thprime}
Let $G(z) \in \mathbb{F}_q[z]^{k\times n}$, with $k \leq n$. The following statements are equivalent:
\begin{enumerate}
  \item The convolutional code with encoder $G(z)$ is non-catastrophic;
  \item the ideal generated by all the $k\times k$ minors of $G(z)$
   is $\mathbb{F}_q[z]$;
  \item for all $u(z) \in \mathbb F_q(z)^k$, $u(z)G(z) \in \mathbb{F}_q[z]^n$ implies that $u(z) \in \mathbb{F}_q[z]^k$, where 
  \[
  \mathbb F_q(z):=\left\{\frac{a(z)}{b(z)}\;:\; a(z),b(z)\in\mathbb{F}_q[z],\;b(z)\neq 0\right\}
  \]
  is the quotient field of $\mathbb F_q[z]$;
  \item the matrix $G(\lambda)$ has full rank for all $\lambda \in \overline{\mathbb F}_q$, where $\overline{\mathbb F}_q$ denotes the algebraic closure of $\mathbb F_q$.
\end{enumerate}
\label{Thm:lprime}
\end{theorem}

Since two encoders of a convolutional code $\cal C$ differ by multiplication on the left by a unimodular matrix, if an encoder of $\cal C$ fulfills any of the conditions of Theorem \ref{Thprime}, then all the encoders of $\cal C$ also fulfill the same conditions.

The correction capability of a code is measured in terms of its distance. An important notion of distance for convolutional codes is the column distance. This notion is defined for convolutional codes that admit a delay-free encoder, i.e., an encoder $G(z)$ with $G(0)$ full row rank. Note that since two encoders of a convolutional code differ by left multiplication of a unimodular matrix, if a convolutional code admits a delay-free encoder, then all its encoders are delay-free and the code is said to be \textbf{delay-free}.

Given a vector $$w(z)=\displaystyle \sum_{i \in \mathbb N_0} w_i z^i \in \mathbb{F}_{q}[z]^n$$ and $j_1,j_2 \in \mathbb N_0$ with $j_1 <j_2$, the vector $${w(z)}_{[j_1,j_2]}=w_{j_1}z^{j_1} + w_{j_1+1}z^{j_1+1} + \cdots + w_{j_2} z^{j_2}$$ is defined as the \textbf{truncation} of $w(z)$ to the interval $[j_1,j_2]$.

\vspace{0.2cm}

\begin{definition}
Let $\cal C$ be an $(n,k,\delta)$ delay-free convolutional code. The \textbf{$j-th$ column distance} of $\cal C$ is defined as
$$
d^c_j({\cal C})  =  \min\{\operatorname{wt}(v(z)_{[0,j]} \, : \; v(z) \in {\cal C} \mbox{ with } v_0 \neq 0\},$$ where $$\operatorname{wt}(v(z)) = \sum_{t \in \mathbb{N}_0} \operatorname{wt} (v_t)$$ is the \textbf{Hamming weight} of $$v(z) = \sum_{t \in \mathbb{N}_0}v_{t}z^{t} \in \mathbb{F}_{q}^n[z]$$ and the weight $\operatorname{wt}(v)$ of $v \in \mathbb{F}_{q}^{n}$ is the number of nonzero components of $v$.
Moreover,
$$
    d_{free}({\cal C})=\min\{{\rm wt}(v(z)) \, : \, v(z) \in {\cal C}\backslash \{0\}\},
    $$
is called \textbf{free distance} of $\cal C$.
\end{definition}

The following inequalities can be directly deduced from the previous definition (see \cite{GRS:sMDS}):
\begin{align}\label{ineq}
0\leq d_0^c(\mathcal{C})\leq d_1^c(\mathcal{C})\leq\cdots\leq d_{free}(\mathcal{C}).    
\end{align}
This implies that $\lim_{j\rightarrow\infty} d_j^c (\mathcal{C})$ always exists.

In \cite{Johannesson2015} it is shown that if one allows messages to be power series, i.e. $u(z)\in\mathbb F_q[[z]]^k$, then $d_{free}(\mathcal{C})= \lim_{j\rightarrow\infty} d_j^c (\mathcal{C})$. In our setting messages have to be polynomials, i.e. $u(z)\in\mathbb F_q[z]^k$. It turns out that in this setting, the equality $d_{free}(\mathcal{C})= \lim_{j\rightarrow\infty} d_j^c (\mathcal{C})$ is only true in general for non-catastrophic convolutional codes. Since it is necessary to modify the proof in \cite{Johannesson2015} when defining convolutional codes as submodules of $\mathbb F_q[z]^n$, we include the proof for the following theorem.

\begin{theorem}\label{Thm:Free-distance-limit}
Let $\mathcal{C}$ be a delay-free convolutional code. If $\mathcal{C}$ is non-catastrophic, then
$\displaystyle d_{free}(\mathcal{C})= \lim_{j\rightarrow\infty} d_j^c (\mathcal{C})$.
\end{theorem}

\begin{proof}
From \eqref{ineq} and the fact that $d_j^c(\mathcal{C})\in\mathbb N$ it follows that there exists $s\in\mathbb N$ such that $$d_s^c(\mathcal{C})=d_{s+1}^c(\mathcal{C})=\cdots=\lim_{j\rightarrow\infty} d_j^c(\mathcal{C}).$$
Let $\nu$ be the memory of $\mathcal{C}$ and $G(z)$ a minimal encoder of $\mathcal{C}$. Take $j\geq \max(s+1,\nu)$ and $v(z)\in\mathcal{C}$ with $v_0\neq 0$ and $wt(v_{[0,j]})=d_j^c(\mathcal{C})=d_s^c(\mathcal{C})$. One obtains $$d_s^c(\mathcal{C})=wt(v_{[0,j]})=wt(v_{[0,s]})+wt(v_{[s+1,j]})\geq d_s^c(\mathcal{C})+wt(v_{[s+1,j]})$$ and it follows $wt(v_{[s+1,j]})=0$ and $wt(v_{[0,s]})=d_s^c(\mathcal{C})$.\\
This implies that for given $u_{j-\nu},\hdots,u_{j-1}\in\mathbb F_q^k$, one finds $u_j\in\mathbb F_q^k$, such that $$\sum_{i=0}^{\nu}u_{j-i}G_i=v_j=0.$$ Since this is true for any $j\geq \max(s+1,\nu)$ and there are only finitely many possibilities for $u_{j-\nu},\hdots,u_{j-1}\in\mathbb F_q^k$, the sequence $(u_i)_{i\in\mathbb N_0}$ obtained in this way will be eventually periodic. Therefore, the corresponding power series $u(z)=\sum_{i\in\mathbb N_0}u_iz^i$ is rational, i.e. $u(z)\in\mathbb F_q(z)^k$. Moreover $v(z)=u(z)G(z)$ fulfills $wt(v(z))=wt(v_{[0,s]})=d_s^c(\mathcal{C})$. As $\mathcal{C}$ was assumed to be non-catastrophic, 
with Theorem \ref{Thm:lprime} we conclude that indeed $u(z)\in\mathbb F_q[z]^k$.
Finally, $d_{free}(\mathcal{C})\leq wt(v(z))=d_s^c(\mathcal{C})$, which completes the proof by \eqref{ineq}.
\end{proof}

The following example shows that the previous theorem is not true in general for catastrophic convolutional codes.

\begin{example}
Consider the catastrophic convolutional code with encoder $G(z)=\arraycolsep=5pt
\begin{bmatrix}
1-z & 1-z
\end{bmatrix}\in\mathbb F_3[z]^{1 \times 2}$. The column distances of this code can be computed as follows. As $u_0\neq 0$, one has for all $j\in\mathbb N_0$ and all $v(z)\in\mathcal{C}$ that $wt(v_{[0,j]})\geq wt(v_0)=2$. Moreover, choosing $\hat{u}(z)=\sum_{i=0}^d z^i$ for some $d\geq j$ results in a codeword $\hat{v}(z)=\hat{u}(z)G(z)$ with $wt(\hat{v}_{[0,j]})=2$, i.e. $d_j^c(\mathcal{C})=2$ for all $j\in\mathbb N_0$.\\
However, any non-zero codeword $v(z)\in\mathcal{C}$ will be of the form $v(z)=u(z)G(z)$ with $u(z)=\sum_{i=r}^su_iz^i$ where $u_r,u_s\neq 0$, $s\geq r$, i.e. $v(z)=\sum_{i=r}^{s+1}v_iz^i$ with $v_r=(u_r\ \ u_r)$ and $v_{s+1}=-(u_s\ \ u_s)$ and $wt(v(z))\geq wt(v_{s+1})+wt(v_r)=4$. Consequently, $d_{free}(\mathcal{C})=4>\lim_{j\rightarrow\infty} d_j^c(\mathcal{C})$.\\
If one would allow power series as inputs as it is done in \cite{Johannesson2015}, $u(z)=\sum_{i=0}^{\infty}z^i$ results in $v(z)=u(z)G(z)=(1\ \ 1)$, i.e. $wt(v(z))=d_{free}(\mathcal{C})=2=\lim_{j\rightarrow\infty} d_j^c(\mathcal{C})$.
\end{example}

The next theorem gives an upper bound for the $j-th$ column distance of a convolutional code $\mathcal{C}$ and shows that if $d_j^c(\mathcal{C})$ is maximal, then the same holds for $d_i^c(\mathcal{C})$ for all $i\leq j$.

\begin{theorem}[\hspace{-0.1mm}\cite{GRS:sMDS}]\label{Thij}
Let $\cal C$ be an $(n,k,\delta)$ delay-free convolutional code. For $j \in \mathbb N_0$,
$$
d^c_j ({\cal C})\leq (n-k)(j+1)+1.
$$
Moreover, if $d^c_j ({\cal C})= (n-k)(j+1)+1$ for some $j \in \mathbb N_0$, then $d^c_i({\cal C}) = (n-k)(i+1)+1$, for $i \leq j$.
\end{theorem}

The free distance of an $(n,k,\delta)$ convolutional code is upper bounded as follows.

\begin{theorem}[\hspace{-0.1mm}\cite{RS:SingB}]
Let $\cal C$ be an $(n,k,\delta)$ convolutional code. Then
$$
d_{free}({\cal C}) \leq  (n-k)\left(\left\lfloor \frac{\delta}{k}   \right\rfloor+1\right) + \delta+1
$$
\end{theorem}
Thus, from \eqref{ineq} if $\mathcal{C}$ is an $(n,k,\delta)$ convolutional code, then
$$
L=\left\lfloor \frac{\delta}{k} \right\rfloor + \left\lfloor \frac{\delta}{n-k}  \right\rfloor
$$
is the largest integer $j$ such that 
the upper bound for $d_j^c$ from Theorem \ref{Thij}
 can be reached 
 (see \cite{GRS:sMDS}).

\begin{definition} An $(n,k,\delta)$ delay-free convolutional code with $d^c_j= (n-k)(j+1)+1$ for all $j \leq L$  is called a \textbf{Maximum Distance Profile (MDP)} convolutional code.
\end{definition}

MDP convolutional codes can be characterized by its encoders as follows. Let $G(z)=\displaystyle \sum_{i =0}^{\nu} G_i z^i \in \mathbb{F}_{q}[z]^{k \times n}$ be an encoder of an $(n,k,\delta)$ delay-free convolutional code of memory $\nu$. 
For $j \in \mathbb N_0$, define the matrix
\begin{equation}\label{Gj}
{\cal G}_j=\left[\begin{array}{cccc}
     G_0 & G_1 & \cdots & G_j \\
      & G_0 & \cdots & G_{j-1} \\
      & & \ddots & \vdots \\
      & & & G_0
\end{array}
\right] \in \mathbb{F}_{q}^{(j+1)k \times (j+1)n},
\end{equation}
where $G_j=0$ for $j>\nu$. This matrix is called the $j$-$th$ \textbf{sliding encoder matrix}.

\vspace{0.2cm}

\begin{theorem}[\hspace{-0.1mm}\cite{GRS:sMDS}] \label{check}
Let $\cal C$ be an $(n,k,\delta)$ delay-free convolutional code,\\ $G(z)=\displaystyle \sum_{i \in \mathbb N_0} G_i z^i \in \mathbb{F}_{q}[z]^{k \times n}$ an encoder of $\cal C$ and ${\cal G}_j$ be defined as in (\ref{Gj}). Then the following are equivalent:
\begin{enumerate}
    \item $d^c_j = (n-k)(j+1)+1$;
    \item every $(j+1)k \times (j+1)k$ full-size minor of ${\cal G}_j$ formed from the columns with indices $ 1 \leq t_1 < \cdots < t_{(j+1)k}$, where $t_{sk+1} > sn$, for $s=1,2, \dots, j$, is nonzero.
\end{enumerate}
\end{theorem}

Note that if $w(z)=u(z)G(z)$ with $u(z)=\sum_{i \in \mathbb N_0} u_i z^i $ and $w(z)=\sum_{i \in \mathbb N_0} w_i z^i$ then
$$
\left[\begin{array}{cccc}
     w_0 & w_1 & \cdots & w_j
\end{array}
\right]=\left[\begin{array}{cccc}
     u_0 & u_1 & \cdots & u_j
\end{array}
\right]{\cal G}_j.
$$

Therefore, if the received codeword contains erasures after transmission, the information sequence can be recovered by considering the equations corresponding to the correct symbols. 
Let $\bar{E}$ be the set of indices in $\{1,\dots, (j+1)n\}$ such that the corresponding components of $\left[\begin{array}{cccc}
     w_0 & w_1 & \cdots & w_j
\end{array}
\right]$ are not erasures, $\mathcal{G}_j^r$ be the submatrix of ${\cal G}_j$ with columns with indices in $\bar E$ and $\left[\begin{array}{cccc}
     w_0 & w_1 & \cdots & w_j
\end{array}
\right]^r$ be the subvector of $\left[\begin{array}{cccc}
     w_0 & w_1 & \cdots & w_j
\end{array}
\right]$ with indices in $\bar E$. Then
$$
\left[\begin{array}{cccc}
     w_0 & w_1 & \cdots & w_j
\end{array}
\right]^r=\left[\begin{array}{cccc}
     u_0 & u_1 & \cdots & u_j
\end{array}
\right]{\cal G}_j^r.
$$

If the coefficient matrix $\mathcal{G}_j^r$ is full row rank, the vector $\left[\begin{array}{cccc} u_0 & u_1 & \cdots & u_j \end{array} \right]$ can be successfully obtained.

If this is not the case, but none of the first $k$ rows of $\mathcal{G}_j^r$ lies in the span of the other rows of $\mathcal{G}_j^r$, then at least $u_0$ can be recovered, and we can proceed by considering the next window $\left[\begin{array}{cccc} u_1 & u_2 & \cdots & u_{j+1} \end{array} \right]$.

Thus, the following lemma immediately follows.

\vspace{0.2cm}

\begin{lemma}[\hspace{-0.1mm}\cite{TRS:Decod}]
    Let $\mathcal{C}$ be an $(n,k,\delta)$ MDP convolutional code. If in any sliding window of length $(j_0+1)n$ at most $(j_0+1)(n-k)$ erasures occur, with $j_0\in\{0,1,\dots,L\}$, then we can completely recover the transmitted sequence.
\end{lemma}



Note that if $G(z)$ is an encoder of memory $\nu$ and $w(z)=u(z)G(z)$ is the transmitted codeword, with $u(z)=\sum_{i \in \mathbb N_0} u_i z^i $, $w(z)=\sum_{i \in \mathbb N_0} w_i z^i$ and $G(z)=\sum_{i =0}^{\nu} G_i z^i$, to recover $u_0$ we consider equations in
\begin{equation}\label{recu0}
\left[\begin{array}{cccc}
     w_0 & w_1 & \cdots & w_{\nu}
\end{array}
\right]=\left[\begin{array}{cccc}
     u_0 & u_1 & \cdots & u_{\nu}
\end{array}
\right]{\cal G}_j,
\end{equation}
since $w_j$ does not depend of $u_0$, for $j > \nu$. Thus, if $\left[\begin{array}{cccc}
     w_0 & w_1 & \cdots & w_{\nu}
\end{array}
\right]$ has at least $n(\nu+1)-k+1$ erasures, then the number of equations in (\ref{recu0}) that we can use to recover $u_0$ is less than $k$ and consequently it will not be possible to obtain $u_0$.

\section{Pseudo-MDP $(n,k,\delta)$ Convolutional Codes}

In this section we will consider $(n,k,\delta)$ convolutional codes, such that $k | \delta$,  with optimal correction capabilities when considering windows up to size $(\nu+1)n$ for $\nu=\frac{\delta}{k}-1$ and with capability to correct burst of erasures of a certain length.   
For the decoding we use a minimal encoder and proceed as described in the previous section.

\begin{definition}
Let $G(z)$ be a minimal encoder of the $(n,k,\delta)$ convolutional code $\cal C$. Let $u(z)=\sum_{i \in \mathbb N} u_i z^i \in \mathbb{F}_{q}[z]^k$ be a message and let $w(z)=\sum_{i \in \mathbb N} w_i z^i \in \mathbb{F}_{q}[z]^n = u(z)G(z)$ be the corresponding codeword.

    $\mathcal{C}$ is said to be able to correct bursts of length $(b+1)n$ with \textbf{burst delay} $d$ and \textbf{symbol delay} t if for any 
 $w(z)$ where for some $i\in\mathbb N_0$, $w_i,\hdots, w_{i+b}$ are erased completely, i.e. each of these vectors has all $n$ components erased, there exists $e\in\mathbb N_0$ such that if each of the vectors $w_{i+b+1}, \hdots,w_{i+b+d}$ contains at most $e$ erasures, 
then $u_{i+j}$ can be recovered with the knowledge of $w_{i+b+1},\hdots, w_{i+\min\{j+t,b+d\}}$ for all $j\in\{0,\hdots, b+d\}$.

\end{definition}

\begin{definition}
    Let $\cal C$ be an $(n,k,\delta)$ convolutional code, such that $k | \delta$ and $n \geq \delta$, and set $\nu :=\frac{\delta}{k}-1$. $\cal C$ is said to be a \textbf{Pseudo-MDP} convolutional code if the following two conditions are fulfilled.
    \begin{enumerate}
        \item  $\cal C$ has optimal $\nu$-th column distance, i.e., $d_{\nu}^c(\mathcal{C})=(n-k)(\nu+1)+1$.
        \item $\cal C$ is able to correct bursts of length $(\nu+1)n$ with burst delay $2$ and symbol delay $\nu+2$ with $e=n-\delta$.
    \end{enumerate} 
\end{definition}





Pseudo-MDP $(n,k, \delta)$ convolutional codes can be obtained from MDP $(n,k, \delta-k)$ convolutional codes as the following theorem shows.

\begin{theorem}\label{th1}
    \label{defPseudo} 
   Let $G(z)=\sum_{i=0}^{\nu}G_iz ^i$ be a minimal encoder of an MDP $(n,k, \delta-k)$ convolutional code $\cal C$, where $k | \delta$ and $n\geq\delta$. This implies $ \nu =\frac{\delta}{k}-1$. Consider 
   $$\bar G(z)=\sum_{i=0}^{\nu}G_iz ^i + G_{\ell}z^{\nu+1},$$
where $\ell \in \{0,1, \dots, \nu\}$. This matrix is an encoder of an $(n,k,\delta)$ Pseudo-MDP convolutional code $\bar{\cal C}$. 
\end{theorem}

\begin{proof}
Obviously, $$d_j^c(\bar{\mathcal{C}})=d_j^c(\mathcal{C})=(n-k)(j+1)+1$$ for $j=0,\hdots,\nu$.
%
%
%
Let us assume now that for some $i \in \mathbb N_0$, $w_{i}, \dots, w_{i+\nu}$ are completely erased and $ w_{i+\nu+1}$ and $ w_{i+\nu+2}$ contain at most $n-\delta$ erasures. One has
\begin{align*}
    w_{i+\nu+1} &=
    \begin{pmatrix}
        u_{i+1} & \cdots & u_{i+\nu - \ell} &
        u_i + u_{i+\nu - \ell+1} & u_{i+\nu - \ell+2} & \cdots & u_{i+\nu+1}
    \end{pmatrix}
    \begin{pmatrix}
        G_{\nu} \\ G_{\nu-1} \\ \vdots \\ G_0
    \end{pmatrix} \\
    w_{i+\nu+2} &=
   \begin{pmatrix}
        u_{i+2} & \cdots & u_{i+\nu - \ell+1} &
        u_{i+1} + u_{i+\nu - \ell+2} & u_{i+\nu - \ell+3} & \cdots & u_{i+\nu+2}
    \end{pmatrix}
    \begin{pmatrix}
        G_{\nu} \\ G_{\nu-1} \\ \vdots \\ G_0
    \end{pmatrix}
\end{align*}

Since $\mathcal{C}$ is MDP with $L=\lfloor\frac{\delta-k}{k}\rfloor+\lfloor\frac{\delta-k}{n-k}\rfloor\geq\nu$, we know from Theorem \ref{check} that all fullsize minors of $\begin{pmatrix}
        G_{\nu} \\ G_{\nu-1}\\ \vdots\\ G_0
    \end{pmatrix}$ are nonzero, i.e. up to $n-(\nu+1)k=n-\delta$ erasures in $w_{i+\nu+1}$ and in $w_{i+\nu+2}$ can be recovered. Hence, we can recover 
    \begin{align*}&u_{i+1}, \dots, u_{i+\nu - \ell},
        u_i + u_{i+\nu - \ell+1}, u_{i+\nu - \ell+2}, \dots, u_{i+\nu+1}, u_{i+\nu - \ell+1},  u_{i+1} + u_{i+\nu - \ell+2},\\
        &u_{i+\nu - \ell+3},\hdots, u_{i+\nu+2} 
        \end{align*}
        which is equivalent to recovering $u_i, u_{i+1}, u_{i+2},\hdots,u_{i+\nu+1},u_{i+\nu+2}$. 
\end{proof}

The following remark shows that it is not possible to recover the given burst of erasures with the encoder $G(z)$ of the $(n,k,\delta-k)$ MDP convolutional code.

\begin{remark}
Note that since $G(z)$ has memory $\nu$, when using $\mathcal{C}$, $w_j$ for $j\geq i+\nu+1$ does not depend on $u_i$ and thus, if $w_i,\hdots, w_{i+\nu}$ are fully erased, $u_i$ cannot be recovered with $G(z)$.
\end{remark}

The following example illustrates how bursts of erasures can be corrected with Pseudo-MDP convolutional codes.

\begin{example}\label{ex1}
Consider the $(3,1,2)$ convolutional code $\cal C$ defined over the field $\mathbb F_{29}$ with encoder
    $$
    G(z)=\begin{pmatrix}
        1 & 2 & 3
    \end{pmatrix} +
    \begin{pmatrix}
       4 & 5 & 6
    \end{pmatrix} z+
    \begin{pmatrix}
        1 & 2 & 7
    \end{pmatrix} z^2.
    $$
    It can be confirmed that $\cal C$ is MDP by Theorem \ref{check}.
    Let us now consider the $(3,1,3)$ Pseudo-MDP convolutional code $\bar{\cal C}$ over the same field with encoder \begin{eqnarray*}\bar G(z) & = & G(z) +  \begin{pmatrix}
        1 & 2 & 7
    \end{pmatrix} z^3
\\& = &\begin{pmatrix}
        1 & 2 & 3
    \end{pmatrix} +
    \begin{pmatrix}
       4 & 5 & 6
    \end{pmatrix} z+
    \begin{pmatrix}
        1 & 2 & 7
    \end{pmatrix} z^2 +  \begin{pmatrix}
        1 & 2 & 7
    \end{pmatrix} z^3.
   \end{eqnarray*}
Let us assume that a codeword $w(z)=\sum_{i \in \mathbb N_0} w_i z^i$ is transmitted and we receive up to time instant $t=4$, $w_0, w_1$ and $w_2$ completely erased and $w_3=\begin{pmatrix}
        9 & 9 & 8
    \end{pmatrix}$ and $w_4=\begin{pmatrix}
        6 & 4 & 5
    \end{pmatrix}$. Thus, since
    $$
   w_3=\begin{pmatrix}
        u_0+u_1 & u_2 & u_3
    \end{pmatrix}
    \begin{pmatrix}
       1 & 2 & 7 \\ 4 & 5 & 6 \\ 1 & 2 & 3
    \end{pmatrix}
    $$
    and $$
     w_4=\begin{pmatrix}
        u_1+u_2 & u_3 & u_4
    \end{pmatrix}
    \begin{pmatrix}
       1 & 2 & 7 \\ 4 & 5 & 6 \\ 1 & 2 & 3
    \end{pmatrix}
    $$
    we obtain
    \begin{eqnarray*}
   \begin{pmatrix}
        u_0+u_1 & u_2 & u_3
    \end{pmatrix} &=&\begin{pmatrix}
          9 & 9 & 8
    \end{pmatrix}\begin{pmatrix}
       1 & 2 & 7 \\ 4 & 5 & 6 \\ 1 & 2 & 3
    \end{pmatrix}^{-1}\\
     & = & \begin{pmatrix}
        3 &  3 & 7
    \end{pmatrix}
    \end{eqnarray*}
    and
    \begin{eqnarray*}
   \begin{pmatrix}
       u_1+u_2 & u_3 & u_4
    \end{pmatrix} &=&\begin{pmatrix}
       6 & 4 & 5
    \end{pmatrix}\begin{pmatrix}
       1 & 2 & 7 \\ 4 & 5 & 6 \\ 1 & 2 & 3
    \end{pmatrix}^{-1}\\
     & = & \begin{pmatrix}
        4 & 7 & 0
    \end{pmatrix}
    \end{eqnarray*}
    and consequently $u_0=2$, $u_1=1$, $u_2=3$, $u_3=7$ and $u_4=0$.
\end{example}

\begin{remark}
Note that in Theorem \ref{th1} it is only required that the convolutional code $\cal C$ has optimal $j-th$ column distances for $j \leq \nu$. These codes are called $\nu$-MDP convolutional codes and were defined in \cite{AL:jMDP}. Thus, this kind of constructions could also be defined for $\nu$-MDP convolutional codes.
\end{remark}

In Theorem \ref{th1}, Pseudo-MDP convolutional codes are constructed from MDP convolutional codes. However, they  may not be MDP and it can even happen that $d_{\nu+1}^c(\bar{\mathcal{C}})$ is smaller than $d_{\nu+1}^c(\mathcal{C})$.



In fact, as the following theorem shows, for $n > \delta$,
 Pseudo-MDP convolutional codes as constructed in Theorem \ref{th1} cannot be MDP.

\begin{theorem}
Let $G(z) = \sum_{i=0}^{\nu} G_i z^i$ be an encoder of an MDP $(n, k, \delta-k)$ convolutional code $\mathcal{C}$, where $\delta = (\nu+1) k$. Let
\[
\bar{G}(z) = \sum_{i=0}^{\nu} G_i z^i + G_{\ell} z^{\nu+1},
\]
with $\ell \in \{0,1,\dots, \nu\}$,
be an encoder of a Pseudo-MDP $(n,k,\delta)$ convolutional code $\bar{\mathcal{C}}$. If $$n > (\nu + 1)k,$$ then $\bar{\mathcal{C}}$ is not an MDP convolutional code.
\end{theorem}


\begin{proof} Let us consider the $(n,k,\delta-k)$ convolutional code $\mathcal{C}$. If $n > (\nu + 1)k$, i.e., $ n > \delta $, then
\[
L = \nu + \left\lfloor \frac{\delta-k}{n - k} \right\rfloor
\]
is equal to $\nu$.

Since $\mathcal{C}$ is an MDP convolutional code, we have $d_j^c = (n - k)(j + 1) + 1$ for $j\leq\nu$. For $j\leq\nu$, the $j-$column distances of $\mathcal{C}$ and of $\bar{\mathcal{C}}$ are the same since 
the corresponding $\nu$-th sliding encoder matrices, ${\cal G}_\nu$ and $\bar{{\cal G}}_\nu$, are equal.

Considering the $(n,k,\delta)$ convolutional code $\bar{\mathcal{C}}$, one observes that
\[
\bar{L} = \left\lfloor \frac{\delta}{k} \right\rfloor + \left\lfloor \frac{\delta}{n - k} \right\rfloor = \left\lfloor \frac{\nu k + k}{k} \right\rfloor + \left\lfloor \frac{\delta }{n - k} \right\rfloor = \nu + 1 + \left\lfloor \frac{\delta}{n - k} \right\rfloor
\] and then conclude that $\bar{L} \geq \nu+1.$

For \( j = \nu + 1 \), we have
\begin{align*}
&\operatorname{wt} \left( \begin{bmatrix}
     u_0 & u_1 & \cdots & u_{\nu+1}
\end{bmatrix} \bar{\mathcal{G}}_{\nu+1} \right)\\
&=
\operatorname{wt} \left(
\begin{bmatrix}
u_0 & u_1 & \dots & u_{\nu} & u_{\nu+1}
\end{bmatrix}
\begin{bmatrix}
G_0 & G_1 & \dots & G_{\nu} & G_{\ell} \\
0 & G_0 & \dots & G_{\nu-1} & G_{\nu} \\
\vdots & \vdots & \ddots & \vdots & \vdots \\
0 & 0 & \dots & G_0 & G_1 \\
0 & 0 & \dots & 0 & G_0
\end{bmatrix}
\right)
\end{align*}

\footnotesize{
\[
= \operatorname{wt} \left(
\begin{bmatrix}
u_0 G_0 
& \dots & u_0 G_{\nu} + u_1 G_{\nu-1} + \cdots + u_{\nu} G_0 & u_1 G_\nu  + \cdots + (u_0+u_{\nu+1-\ell})G_{\ell}  + \cdots + u_{\nu+1} G_0
\end{bmatrix}
\right).
\]}

\normalsize

Let $u_0$ be such that \( \operatorname{wt}(u_0 G_0) = n - k + 1 \), \( u_{\nu+1-\ell} = -u_0 \) and $u_i=0$ for $i \in \{0,1,\dots, \nu+1\} \backslash \{0,\nu+1-\ell\}$. Then, since $ u_{\nu+1-\ell} = -u_0$,


\begin{align}
 & \operatorname{wt}\left(
\begin{bmatrix}
u_0 & 0 & \dots & 0 & u_{\nu+1-\ell} & 0 & \dots & 0
\end{bmatrix}
\begin{bmatrix}
G_0 & G_1 & \dots & G_{\nu} & G_{\ell} \\
0 & G_0 & \dots & G_{\nu-1} & G_{\nu} \\
\vdots & \vdots & \ddots & \vdots & \vdots \\
0 & 0 & \dots & G_0 & G_1 \\
0 & 0 & \dots & 0 & G_0
\end{bmatrix}
\right) \nonumber\\
&= \operatorname{wt} \left(
\begin{bmatrix}
u_0 G_0 & u_0 G_1 & \dots & u_0 G_{\nu - \ell} &  & u_0(G_{\nu - \ell+1}-G_0) & \cdots & u_0 (G_\nu-G_{\ell-1}) & 0 \nonumber
\end{bmatrix}
\right) \\
& \leq (n-k+1) + \nu \cdot n= n(\nu + 1) - k + 1.\nonumber
\end{align} This allows us to conclude that $\bar{\mathcal{C}}$ is not MDP, as $$
\operatorname{wt}(u(z) G(z)) \leq n(\nu + 1) - k + 1 < (n - k)(\nu + 2) + 1.
$$
\end{proof}

The theorem above only covers the case where $n > \delta$. The complementary case where $n =\delta $ is less straightforward. Indeed, if $n = \delta$ and $\mathcal{C}$ is an MDP convolutional code, then the Pseudo-MDP convolutional code $\bar{\mathcal{C}}$ may or may not be MDP as the following examples show.

\begin{example}
 Consider the convolutional codes $\cal C$ and $\bar{\cal C}$ and the respective encoders given in Example \ref{ex1}.
    Since $\cal C$ is MDP, $d_3^c({\cal C})=9$, but $d_3^c(\bar{\cal C}) <9$ because, for example,
    \begin{align*}
    {\big(2+ 25z + 14z^2 + 25z^3\big)\bar{G}(z)}_{[0,3]} \\
    = \begin{bmatrix}
        2 + 4z + 21z^3 &
        4 + 2z + 12z^2 &
        6 + 3z^2
    \end{bmatrix}
\end{align*}
which has weight $8$.
\end{example}

\begin{example}\label{ex3} 
 The $(3,1,2)$ convolutional code $\mathcal{C}$
over $\mathbb{F}_{53}$ with encoder
$$
    G(z)=\begin{pmatrix}
        52 & 49 & 30
    \end{pmatrix} +
    \begin{pmatrix}
       16 & 29 & 14
    \end{pmatrix} z+
    \begin{pmatrix}
        22 & 45 & 41
    \end{pmatrix} z^2,
    $$
has column distances $d_0^c = 3$, $d_1^c = 5$, $d_2^c = 7$ and $d_3^c = 9$, meaning that $\mathcal{C}$ is an MDP convolutional code. In addition, the $(3,1,3)$ convolutional code $\bar{\mathcal{C}}$ with encoder
\begin{eqnarray*}\bar G(z) & = & G(z) +  \begin{pmatrix}
        22 & 45 & 41
    \end{pmatrix}z^3\\& = & \begin{pmatrix}
        52 & 49 & 30
    \end{pmatrix} + \begin{pmatrix}
       16 & 29 & 14
    \end{pmatrix} z+
    \begin{pmatrix}
        22 & 45 & 41
    \end{pmatrix} z^2 +  \begin{pmatrix}
        22 & 45 & 41
    \end{pmatrix} z^3,
   \end{eqnarray*}
has column distances $d_0^c = 3$, $d_1^c = 5$, $d_2^c= 7$, $d_3^c = 9$ and $d_4^c = 11$, meaning that $\bar{\mathcal{C}}$ is also an MDP convolutional code.
\end{example}

Note that if an $(n,k,\delta)$ convolutional code $\cal C$, with $k \mid \delta$, is MDP then it is non-catastrophic \cite{AlfLieb20}. This implies that the Pseudo-MDP convolutional of Example \ref{ex3} is also non-catastrophic. The Pseudo-MDP convolutional code of Example \ref{ex1} is not MDP but is still non-catastrophic (this could be seen by checking that the minors of 
$\bar G(z)$ are coprime). 

Moreover, an $(n,1,1)$ convolutional code with encoder
$$
G(z)=\begin{pmatrix}
        1 & 1 & \cdots & 1 
    \end{pmatrix} +
    \begin{pmatrix}
        a_1 & a_2 & \cdots & a_n
    \end{pmatrix}z 
$$
is MDP if and only $a_i \neq a_j$, for $i \neq j$. In this case, if we consider
$$
\bar G(z)=\begin{pmatrix}
        1 & 1 & \cdots & 1 
    \end{pmatrix} +
    \begin{pmatrix}
        a_1 & a_2 & \cdots & a_n
    \end{pmatrix}z + \begin{pmatrix}
        a_1 & a_2 & \cdots & a_n
    \end{pmatrix}z^2 
$$
or
$$
\tilde G(z)=\begin{pmatrix}
        1 & 1 & \cdots & 1 
    \end{pmatrix} +
    \begin{pmatrix}
        a_1 & a_2 & \cdots & a_n
    \end{pmatrix}z + \begin{pmatrix}
        1 & 1 & \cdots & 1 
    \end{pmatrix}z^2
$$
we always obtain non-catastrophic Pseudo-MDP convolutional codes. Note that column multiplication of an encoder by a nonzero constant does change the properties to be MDP or non-catastrophic. Hence Pseudo-MDP convolutional codes obtained from  $(n,1,1)$ MDP convolutional codes are always non-catastrophic.

However if an MDP convolutional code has degree $0$ (i.e., it is an MDS block code) then the corresponding Pseudo-MDP code is always catastrophic. For other code parameters the existence of catastrophic Pseudo-MDP convolutional codes is not known.



So far, we presented constructions for $(n,k,\delta)$ Pseudo-MDP convolutional codes for all $k | \delta$ and $n\geq(\nu+1)k$ where $\nu=\frac{\delta}{k}$. 
In the following, we give conditions and an example where the condition $n\geq(\nu+1)k$ can be weakened.

If $n\geq \frac{\nu+3}{2} k$ and all fullsize minors of

$$  \begin{pmatrix}
        G_{\ell} & 0\\ G_{\nu} & G_{\ell} \\ G_{\nu-1} & G_{\nu} \\ \vdots & G_{\nu-1} \\ G_0 & \vdots\\ 0 & G_0
    \end{pmatrix}$$
are nonzero, then we can use the equation
\begin{align*}
  \begin{pmatrix}
        u_{i} & \cdots & u_{i+\nu +2}
    \end{pmatrix}  \begin{pmatrix}
        G_{\ell} & 0\\ G_{\nu} & G_{\ell} \\ G_{\nu-1} & G_{\nu} \\ \vdots & G_{\nu-1} \\ G_0 & \vdots\\ 0 & G_0
    \end{pmatrix} = (w_{i+\nu+1} \ w_{i+\nu+2} )
\end{align*}
to recover $\begin{pmatrix}
        u_{i} & \cdots & u_{i+\nu +2}
    \end{pmatrix}$ in case that there are not more than $n- \frac{\nu+3}{2} k$ erasures in $(w_{i+\nu+1} \ w_{i+\nu+2} )$. In this way, the bound on $n$ is lowered from $(\nu+1)k$ to $\frac{\nu+3}{2}k$ and we can tolerate $n- \frac{\nu+3}{2} k$ erasures in $(w_{i+\nu+1} \ w_{i+\nu+2} )$ instead of $n- (\nu+1)k$ erasures in each of the vectors $w_{i+\nu+1}$ and $w_{i+\nu+2}$.
Note that we can tolerate more erasures with this new condition in case $2(n-(\nu+1)k)\leq n-\frac{\nu+3}{2}k$ which is equivalent to $n\leq k(\frac{3}{2}\nu+\frac{1}{2})$.

One can check that the encoder $\bar{G}(z)$ of Example \ref{ex3} fulfills the condition that all fullsize minors of  
$$  \begin{pmatrix}
        G_{\ell} & 0\\ G_{\nu} & G_{\ell} \\ G_{\nu-1} & G_{\nu} \\ \vdots & G_{\nu-1} \\ G_0 & \vdots\\ 0 & G_0
    \end{pmatrix}$$ are nonzero.
But in this case, the condition $n\geq \frac{\nu+3}{2}k=2.5k=2.5$ does not give an improvement over the condition $n\geq(\nu+1)k=3k=3$ from Theorem \ref{th1} and also the number of erasures that can be tolerated stays the same (actually it is equal to zero in both cases).

However, we can use Example \ref{ex3} to obtain the following improved construction.
For any $n,k$ with $n\geq 3k$, assume that $G(z)=G_0+G_1z+G_2z^2+G_3z^3$ is an MDP convolutional code.
Via row operations we obtain

$$
 \begin{pmatrix}
        G_{1} & 0\\ G_{3} & G_{1} \\ G_{2} & G_{3} \\ G_1 & G_{2} \\ G_0 & G_1\\ 0 & G_0
    \end{pmatrix}
    \sim 
\begin{pmatrix}
    G_3-G_0 & 0\\
    G_1 & 0\\
       G_2 & G_3-G_0\\
    G_0 & G_1\\
    0 & G_2\\
    0 & G_0
\end{pmatrix}.$$
If all fullsize minors of this matrix
are nonzero, then $\bar{G}(z)=G_0+G_1z+G_2z^2+G_3z^3+G_1z^4$ can correct the burst
if not more than $n-3k$ erasures happen in $(w_{i+\nu}, w_{i+\nu+1})$. Note that in Theorem \ref{th1}, for this case it would be required to have $n\geq 4k$ and for $n\leq 5k$, we can tolerate more erasures with the improved construction. 

\begin{example}\label{ex:4}
Using for $G(z)$ the Pseudo-MDP code of Example \ref{ex3}, which is also MDP, and extending it via $\bar{G}(z)=G_0+G_1z+G_2z^2+G_2z^3+G_1z^4$
we can do further row operations to obtain

$$
\begin{pmatrix}
    G_2-G_0 & 0\\
    G_1 & 0\\
       G_2 & G_2-G_0\\
    G_0 & G_1\\
    0 & G_2\\
    0 & G_0
\end{pmatrix}
\sim \begin{pmatrix}
    G_2 & 0\\
    G_1 & 0\\
       G_0 & 0\\
    0 & G_2\\
    0 & G_1\\
    0 & G_0
\end{pmatrix}
\in\mathbb{F}_{q}^{6\times 6}.$$
The determinant of this matrix is nonzero and hence, we obtain a $(3,1,4)$ Pseudo-MDP code with $n<\delta=(\nu+1)k$, which we could not obtain with the previous theorems.
\end{example}

In the following section, we will generalize these ideas to obtain codes with good column distances which are able to correct certain bursts and allow constructions for higher code rates than the previous constructions for Pseudo-MDP codes.
In the previous example we used a Pseudo-MDP code that was also MDP as basis for the construction. Since it is not easy to find such codes, for the general idea we relax this condition, i.e. in turn for asking for higher code rates, we will relax the conditions on the column distances. Asking fewer column distances to be optimal has the additional advantage that the MDP code we will start with and extend afterwards has smaller code parameters and hence, can be constructed over smaller finite fields.

\section{$x$-Pseudo MDP Convolutional Codes}

In the following we will generalize the definition of Pseudo-MDP convolutional codes relaxing the conditions on the column distances in order to enable constructions over smaller fields and for larger code rates.

\begin{definition}
    Let $\cal C$ be an $(n,k,\delta)$ convolutional code, such that $k | \delta$ and $n \geq \delta-(x-1)k$, and set $\nu :=\frac{\delta}{k}-x$ for some $x\in\mathbb N$. $\cal C$ is said to be an $x$-Pseudo-MDP convolutional code if the following two conditions are fulfilled.
    \begin{enumerate}
        \item  $\cal C$ has optimal $\nu$-th column distance, i.e., $d_{\nu}^c(\mathcal{C})=(n-k)(\nu+1)+1$.
        \item $\cal C$ is able to correct bursts of length $(\nu+x)n$ with burst delay $x+1$ and symbol delay $\nu+1+x$ with $e=n-\delta+(x-1)k$.
    \end{enumerate} 
\end{definition}

Note that the length of the correctable burst in condition (2) is equal to $\frac{\delta}{k}n$, which is the same as for the Pseudo-MDP convolutional codes from Definition \ref{defPseudo}.
However, with increasing $x$, we can tolerate more erasures in each part of the $x+1$ windows after the burst, but in turn we have to tolerate an increased delay. The main advantage of increasing $x$ is, as mentioned before the previous definition, that the field size can be decreased and the code rate can be increased.


In the following, we present two different constructions for $x$-Pseudo MDP convolutional codes. The first one has the advantage that it works for all $\nu\in\mathbb N$. The second one has the advantage that it allows for easier erasure correction, but it works only for $\nu\geq 2$.


\begin{theorem}\label{th2}
For $x\in\mathbb N$, let $G(z)=\sum_{i=0}^{\nu}G_iz ^i$ be a minimal encoder of an MDP $(n,k, \delta-xk)$ convolutional code $\cal C$, where $k | \delta$ and $n\geq\delta-(x-1)k$. This implies $ \nu =\frac{\delta}{k}-x$. Consider $G^x(z)=G(z)+G_{\ell}\cdot\sum_{i=\nu+1}^{\nu+x}z^i$ for $\ell\in\{1,\hdots,\nu\}$.
 This matrix is an encoder of an $(n,k,\delta)$ $x$-Pseudo-MDP convolutional code $\mathcal{C}^x$. 
\end{theorem}

\begin{proof}
Obviously, $$d_j^c(\mathcal{C}^x)=d_j^c(\mathcal{C})=(n-k)(j+1)+1$$ for $j=0,\hdots,\nu$.
Since $\mathcal{C}$ is MDP with $$L=\Big\lfloor\frac{\delta-xk}{k}\ell\Big\rfloor+\Big\lfloor\frac{\delta-xk}{n-k}\Big\rfloor\geq\nu,$$ we know from Theorem \ref{check} that all fullsize minors of $$\begin{pmatrix}
        G_{\nu} \\ G_{\nu-1}\\ \vdots\\ G_0
    \end{pmatrix}$$ are nonzero. This implies that up to $n-(\nu+1)k=n-\delta+(x-1)k$ erasures in each of the vectors $w_{i+\nu+x}$ and $w_{i+\nu+x+1}$ 
    can be recovered using the equations

\begin{align*}
    w_{i+\nu+x} &= 
    \begin{pmatrix}
        u_{i+x} \\ \vdots \\ u_{i+\nu+x-\ell-1} \\  u_i + \cdots + u_{i+x-1}+u_{i+\nu+x-\ell} \\ u_{i+\nu+x-\ell+1} \\ \vdots \\ u_{i+\nu+x}
    \end{pmatrix}^T
    \begin{pmatrix}
        G_{\nu} \\ \vdots \\  G_{\ell}\\ \vdots \\ G_0
    \end{pmatrix} \\
    w_{i+\nu+x+1} &= 
    \begin{pmatrix}
        u_{i+x+1} \\ \vdots \\ u_{i+\nu+x-\ell} \\  u_{i+1} + \cdots + u_{i+x}+u_{i+\nu+x-\ell+1} \\ u_{i+x+\nu-\ell+2} \\ \vdots \\ u_{i+\nu+x+1}
    \end{pmatrix}^T
    \begin{pmatrix}
        G_{\nu} \\ \vdots \\  G_{\ell}\\ \vdots \\ G_0
    \end{pmatrix} 
\end{align*}
to obtain
$$u_{i+x},\hdots,u_{i+\nu+x-\ell-1} ,u_i + \cdots + u_{i+x-1}+u_{i+\nu+x-\ell}, u_{i+\nu+x-\ell+1},\hdots,u_{i+\nu+x}$$
and 
$$u_{i+x+1},\hdots,u_{i+\nu+x-\ell}, u_{i+1} + \cdots + u_{i+x}+u_{i+\nu+x-\ell+1}, u_{i+\nu+x-\ell+2}, \hdots, u_{i+\nu+x+1}.$$
For $\ell=\nu$, we can use this to obtain
$u_i=u_i+\cdots +u_{i+x}-( u_{i+1}+\cdots+u_{i+x+1})+u_{i+x+1}$. If $\ell<\nu$, then
$u_i=u_i+\cdots +u_{i+x-1}+u_{i+\nu+x-\ell}-( u_{i+1}+\cdots+u_{i+x}+u_{i+\nu+x-\ell+1})+u_{i+x}+u_{i+\nu+x-\ell+1}-u_{i+\nu+x-\ell}$ can be recovered as well.
 Also for the recovery of the remaining message vectors, we distinguish the cases $\ell<\nu$ and $\ell=\nu$.

If $\ell<\nu$, after Step 1, i.e. considering $w_{i+\nu+x}$ and $w_{i+\nu+x+1}$, we recovered
$$u_i,u_{i+x},\hdots,u_{i+\nu+x+1}.$$
In Step 2, we consider what we already recovered from $w_{i+\nu+x+1}$ together with

$$ w_{i+\nu+x+2} = 
    \begin{pmatrix}
        u_{i+x+2} \\ \vdots \\ u_{i+\nu+x-\ell+1} \\ u_{i+2} + \cdots + u_{i+x+1}+u_{i+\nu+x-\ell+2} \\ u_{i+x+\nu-\ell+3} \\ \vdots \\ u_{i+\nu+x+2}
    \end{pmatrix}^T
    \begin{pmatrix}
        G_{\nu} \\ \vdots \\  G_{\ell}\\ \vdots \\ G_0
    \end{pmatrix} 
$$
to recover 
$u_{i+1}$ and $u_{i+\nu+x+2}$.
In general, for $j\in\{1,\hdots,x\}$, after step j,
considering 
$w_{i+\nu+x},\hdots,w_{i+\nu+x+j}$, we recovered
$$u_i,\hdots, u_{i+j-1},u_{i+x},\hdots,u_{i+\nu+x+j},$$
meaning that we can recover the whole burst after step $j=x$, i.e. with delay $x+1$.

If $\ell=\nu$, considering
$w_{i+\nu+x}$, one recovers 
$$u_i+\cdots+u_{i+x},u_{i+x+1},\hdots, u_{i+x+\nu}$$
and considering 
$w_{i+\nu+x+1}$, one recovers 
$$u_{i+1}+\cdots+u_{i+x+1},u_{i+x+2},\hdots, u_{i+x+\nu+1},$$
i.e. after Step 1, one knows 
$$u_i, u_{i+1}+\cdots+u_{i+x},u_{i+x+1},\hdots, u_{i+x+\nu+1}.$$
In Step 2, considering $w_{i+\nu+x+2}$, we recover
$$u_{i+2}+\cdots+u_{i+x+2},u_{i+x+3},\hdots, u_{i+x+\nu+2},$$
i.e. after this step, in total we recovered
$$u_i, u_{i+1},u_{i+2}+\cdots+u_{i+x},u_{i+x+1},\hdots, u_{i+x+\nu+2}.$$
In general after Step $j$, considering received vectors up to $w_{i+\nu+x+j}$
 for $j\in\{1,\hdots,x\}$, one recovers
$$u_i, \hdots, u_{i+j-1} ,u_{i+j}+\cdots+u_{i+x},u_{i+x+1},\hdots, u_{i+x+\nu+j},$$
meaning that we can recover the whole burst after step $j=x$, i.e. with delay $x+1$.\\
\end{proof}


\begin{remark} Note that if $w_{i}, \hdots, w_{i+\nu+x-1}$ are completely erased, $u_i$ cannot be recovered using $G(z)$ since $w_j$, for $j \geq i+\nu+x$, does not depend on $u_i$.
\end{remark}

\begin{remark}
If we choose $\ell=0$ in the previous theorem, we can recover\\ $u_{i+x},\hdots,u_{i+x+\nu-1}, u_i+\cdots+u_{i+x-1}+u_{i+x+\nu}$ from $w_{i+\nu+x}$ and $u_{i+x+1},\hdots,u_{i+x+\nu}, u_{i+1}+\cdots+u_{i+x}+u_{i+x+\nu+1}$ from $w_{i+\nu+x+1}$, which only allows us to recover $u_i-u_{i+x+\nu+1}$ but not $u_i$.
\end{remark}


\begin{theorem}
For $x\in\mathbb N$, let $G(z)=\sum_{i=0}^{\nu}G_iz ^i$ be a minimal encoder of an MDP $(n,k, \delta-xk)$ convolutional code $\cal C$, where $k | \delta$ and $n\geq\delta$. This implies $ \nu =\frac{\delta}{k}-x$. Consider  $\bar{G}^x(z)=G(z)+G_{0}\cdot\sum_{i=\nu+1}^{\nu+x-1}z^i + G_{\ell}z^{\nu + x}$, $\ell \in \{1, \dots, \nu-1\}$.
 This matrix is an encoder of an $(n,k,\delta)$ $x$-Pseudo-MDP convolutional code $\bar{\mathcal{C}}^x$. 
\end{theorem}


\begin{proof}
Obviously, $$d_j^c(\bar{\mathcal{C}}^x)=d_j^c(\mathcal{C})=(n-k)(j+1)+1$$ for $j=0,\hdots,\nu$.

For $\bar{G}^x(z)$, we have the equations
\begin{align*}
    w_{i+\nu+x} &=
    \begin{pmatrix}
        u_{i+x} \\ \vdots \\ u_{i+\nu+x-\ell-1} \\ u_i + u_{i+\nu+x- \ell} \\ u_{i+\nu+x-\ell+1} \\ \vdots \\ u_{i+\nu+x-1} \\ u_{i+\nu + x} + u_{i+1} + \cdots + u_{i+x-1}
    \end{pmatrix}^T
    \begin{pmatrix}
        G_{\nu} \\ G_{\nu-1} \\ \vdots \\ G_0
    \end{pmatrix} \\
    w_{i+\nu+x+1} &=
    \begin{pmatrix}
        u_{i+x+1} \\ \vdots \\ u_{i+\nu+x-\ell} \\ u_{i+1} + u_{i+\nu+x- \ell+1} \\ u_{i+\nu+x-\ell+2} \\ \vdots \\ u_{i+\nu+x} \\ u_{i+\nu + x+1} + u_{i+2} + \cdots + u_{i+x}
    \end{pmatrix}^T
    \begin{pmatrix}
        G_{\nu} \\ G_{\nu-1} \\ \vdots \\ G_0
    \end{pmatrix}
\end{align*}
and using that $G(z)$ is the encoder of an MDP convolutional code as in the proof of Theorem \ref{th1} and that $\ell<\nu$, we can recover 
\begin{align*}
& u_{i+x},\hdots, u_{i+\nu+x-\ell}, u_i + u_{i+\nu+x- \ell}, u_{i+1} + u_{i+\nu+x- \ell+1}, u_{i+\nu + x} + u_{i+1} + \cdots + u_{i+x-1},\\ &u_{i+\nu + x+1} + u_{i+2} + \cdots + u_{i+x}, 
u_{i+\nu+x- \ell+1},\hdots, u_{i+\nu+x-1}, u_{i+\nu+x- \ell+2},\hdots, u_{i+\nu+x},  
\end{align*}
which implies that for $\ell>1$, after Step 1, we can recover
\begin{align*}
u_i,u_{i+1},u_{i+2}+\cdots+u_{i+x-1}, u_{i+x},\hdots, u_{i+\nu+x+1}.
\end{align*}
In Step 2, i.e. considering $w_{i+\nu+1}$ and $w_{i+\nu+2}$, we recover
\begin{align*}
u_{i+1},u_{i+2}, u_{i+3}+\cdots+u_{i+x-1}, u_{i+x},\hdots, u_{i+\nu+x+2},
\end{align*}
i.e., after Step 2, we know 
\begin{align*}
u_i,u_{i+1},u_{i+2}, u_{i+3}, u_{i+4}+\cdots+u_{i+x-1}, u_{i+x},\hdots, u_{i+\nu+x+2}.
\end{align*}
In general, after Step $j$, we know
\begin{align*}
u_i,u_{i+1},u_{i+2}, u_{i+2j-1}, u_{i+2j}+\cdots+u_{i+x-1}, u_{i+x},\hdots, u_{i+\nu+x+j}.
\end{align*}
i.e., for $x\geq2$, we recover the whole burst as soon as $2j\geq x-1$, meaning we can recover the whole burst with delay $\lceil\frac{x-1}{2}\rceil+1$, which is even smaller than the required delay of $x+1$.

If $\ell=1$, in the first step we only recover
\begin{align*}
u_i,u_{i+1}+u_{i+\nu+x}, u_{i+x},\hdots, u_{i+\nu+x-1}, u_{i+\nu+x+1}, u_{i+2}+\cdots+u_{i+x-1}.  
\end{align*}
In Step 2, we then recover,
\begin{align*}
u_{i+1},u_{i+2}+u_{i+\nu+x+1}, u_{i+x+1},\hdots, u_{i+\nu+x}, u_{i+\nu+x+2}, u_{i+3}+\cdots+u_{i+x},  
\end{align*}
i.e. after Step 2, we know 
\begin{align*}
u_i,u_{i+1}, u_{i+2}, u_{i+3}+\cdots+u_{i+x-1}, u_{i+x},\hdots,u_{i+x+\nu+2}.  
\end{align*}
In general after Step $j\geq 2$, one knows
\begin{align*}
u_i,u_{i+1}, \hdots, u_{i+j}, u_{i+j+1}+\cdots+u_{i+x-1}, u_{i+x},\hdots,u_{i+x+\nu+j}, 
\end{align*}
i.e. we recover the whole burst as soon as $j\geq x-2$, i.e. with delay $x-1$.
\end{proof}

\begin{remark}
If we would choose $\ell=\nu$ in the previous construction, we would not be able to recover $u_i$ in the first step but we could only recover
\begin{align*}
    &u_i + u_{i+x}, \quad u_{i+1} + u_{i+x+1}, \quad u_{i+x+1}, \dots, u_{i+x+\nu-1}, \quad u_{i+x+2}, \dots, u_{i+x+\nu}, \\
    &u_{i+x+\nu}+u_{i+1}+\cdots+u_{i+x-1}, \quad u_{i+x+\nu+1}+u_{i+2}+\cdots+u_{i+x}
\end{align*}
i.e., we obtain (among others)
$$u_i + u_{i+x}, u_{i+\nu+x+1}+u_{i+x}$$ 
which does not allow the recovery of $u_i$ in the first step. 
\end{remark}

\begin{remark} As happens in Theorem \ref{th2}, if $w_{i}, \hdots, w_{i+\nu+x-1}$ are completely erased, $u_i$ cannot be recovered using $G(z)$.
\end{remark}


\section{Conclusion}
In this paper we propose $x$-Pseudo MDP convolutional code. These codes have fewer maximal column distances than MDP convolutional codes but in turn can recover larger bursts of erasures and can be defined over smaller finite fields. We present constructions of these codes starting with an MDP convolutional codes of smaller degree and repeating some of the coefficient matrices of its encoder.



\section*{Acknowledgment}
The first and third author are supported by the Center for Research and Development in Mathematics and Applications (CIDMA) (\url{https://ror.org/05pm2mw36})
under the Portuguese Foundation for Science and Technology (FCT)\\ (\url{https://ror.org/00snfqn58}), Grants 
UID/04106/2025 (\url{https://doi.org/10.54499/UID/04106/2025})
and
UID/PRR/04106/2025 (\url{https://doi.org/10.54499/UID/PRR/04106/2025}). The second author is supported by the German research foundation, project number 513811367.


\begin{thebibliography}{00}

\bibitem{AlfLieb20}
G. N. Alfarano, J. Lieb, On the left primeness of some polynomial matrices with applications to convolutional codes, {\it Journal of Algebra and Its Applications} {\bf 20(11)} (2021) 2150207.

\bibitem{rscc}
G. N. Alfarano, D. Napp, A. Neri and V. Requena, Weighted Reed--Solomon convolutional codes, {\it Linear Multilinear Algebra} (2023).

\bibitem{AL:jMDP}
P. Almeida and J. Lieb, Complete j-MDP Convolutional Codes, {\it IEEE Trans. Inf. Theory} {\bf 66} (2020) 7348--7359.

\bibitem{ANP:MDP}
P. Almeida, D. Napp and R. Pinto, A new class of superregular matrices and MDP convolutional codes, {\it Linear Algebra Appl.} {\bf 439} (2013) 2145--2157.

\bibitem{chen}
Z. Chen, A Lower Bound on the Field Size of Convolutional Codes With a Maximum Distance Profile and an Improved Construction, {\it IEEE Trans. Inf. Theory} {\bf 70} (2024) 4064--4073.

\bibitem{cheng2026new}
Z. Cheng, New constructions of MDP convolutional codes with memory 1 for $k=2,3$, {\it Cryptography and Communications} (2026) 1--14.

\bibitem{itw}
S. Dang, J. Lieb, O. Makkonen, P. Soto, A. Sprintson, A matrix completion approach for the construction of MDP convolutional codes, {\it Proc. 2025 IEEE Information Theory Workshop (ITW)} (2025) 710--715.

\bibitem{GRS:sMDS}
H. Gluesing-Luerssen, J. Rosenthal and R. Smarandache, Strongly MDS convolutional codes, {\it IEEE Trans. Inf. Theory} {\bf 52} (2006) 584--598.


\bibitem{Johannesson2015}
R. Johannesson, K. Zigangirov, Fundamentals of convolutional coding, 2nd edn. Wiley-IEEE
Press (2015).

\bibitem{kailath1980}
T. Kailath, Linear Systems. Englewood Cliffs, N.J.: Prentice Hall (1980)

\bibitem{L:cMDP}
J. Lieb, Complete MDP convolutional codes, {\it J. Algebra Appl.} {\bf 18} (2019) 1950105.

\bibitem{LPJ:ConvC}
J. Lieb, R. Pinto and J. Rosenthal, Convolutional Codes, in {\it A Concise Encyclopedia of Coding Theory}, W.C. Huffman, J.-L. Kim, and P. Solé (Eds.), Boca Raton, (2021) 197--225.

\bibitem{unitmemorymdp}
G. Luo, X. Cao, M. F. Ezerman and S. Ling, A Construction of Maximum Distance Profile Convolutional Codes With Small Alphabet Sizes, {\it IEEE Trans. Inf. Theory} {\bf 69} (2023) 2983--2990.

\bibitem{MUNOZCASTANEDA}
Á. L. Muñoz Castañeda and F. J. Plaza-Martín, On the existence and construction of maximum distance profile convolutional codes, {\it Finite Fields Appl.} {\bf 75} (2021) 101877.

\bibitem{NS:MDP}
D. Napp and R. Smarandache, Constructing strongly-MDS convolutional codes with maximum distance profile, {\it Adv. Math. Commun.} {\bf 10} (2016) 275--290.

\bibitem{RS:SingB}
J. Rosenthal and R. Smarandache, Maximum distance separable convolutional codes, {\it Appl. Algebra Engrg. Commun. Comput.} {\bf 10} (1999) 15--32.

\bibitem{TRS:Decod}
V. Tomás, J. Rosenthal and R. Smarandache, Decoding of MDP convolutional codes over the erasure channel, {\it Proc. 2009 IEEE Int. Symp. Inf. Theory (ISIT 2009)} (2009) 556--560.

\bibitem{York97}
E.V. York, Algebraic Description and Construction of Error
Correcting Codes: A Linear Systems Point of View. Ph.D. dissertation, University of Notre Dame, 1997.



\end{thebibliography}
\end{document}